\documentclass[12pt]{article}

\usepackage{amssymb,amsthm,amsmath}
\usepackage{algorithm}
\usepackage{algpseudocode}
\usepackage{thmtools,thm-restate}
  

\usepackage{subfig}
\usepackage[usenames,dvipsnames]{xcolor}
\usepackage[colorlinks,citecolor=blue,linkcolor=BrickRed]{hyperref}\usepackage[colorlinks,citecolor=blue,linkcolor=BrickRed]{hyperref}
\usepackage{fullpage}
\usepackage{graphicx}
\usepackage{xspace}
\usepackage{enumitem}
\usepackage{thmtools,thm-restate}
\usepackage{wrapfig}
\usepackage{comment}
\usepackage{cleveref}

\graphicspath{{figures/}}

\newtheorem{theorem}{Theorem}[section]
\newtheorem{corollary}[theorem]{Corollary}
\newtheorem{conjecture}[theorem]{Conjecture}
\newtheorem{lemma}[theorem]{Lemma}

\newtheorem{definition}[theorem]{Definition}

\newif\ifFULL
\FULLtrue

\newcommand{\IGNORE}[1]{}

\usepackage{tikz}
\usetikzlibrary{arrows}
\usetikzlibrary{arrows.meta}
\usetikzlibrary{shapes}
\usetikzlibrary{backgrounds}
\usetikzlibrary{positioning}
\usetikzlibrary{decorations.markings}
\usetikzlibrary{patterns}
\usetikzlibrary{calc}
\usetikzlibrary{fit}
\usetikzlibrary{snakes}
\tikzset{
    >=stealth',
    pil/.style={
           ->,
           thick,
           shorten <=2pt,
           shorten >=2pt,}
}
\tikzset{->-/.style={decoration={
  markings,
  mark=at position .5 with {\arrow{>}}},postaction={decorate}}}

\newcommand{\N}{\mathbb{N}}

\newcommand{\setR}{\mathbb{R}}

\newcommand{\E}{\mathbb{E}}

\newcommand{\eat}[1]{}
\newcommand{\hide}[1]{{\Large \color{red} Contents here are hidden! To reveal contents, remove this command.}}
\newcommand{\inn}[2]{\langle {#1}, {#2} \rangle}

\newcommand{\tr}{\mathsf{tr}}
\newcommand{\diag}{\mathsf{diag}}

\newcommand{\op}{\ensuremath{\mathrm{op}}\xspace}

\newcommand{\vol}{\mathsf{vol}}

\newcommand{\rank}{\mathsf{rank}}

\allowdisplaybreaks

{\hspace*{\fill}$\Box$\par}

\title{A New Framework for Matrix Discrepancy:\\ Partial Coloring Bounds via Mirror Descent}

\author{Daniel Dadush\thanks{Centrum Wiskunde \& Informatica, Amsterdam, Netherlands. \texttt{dadush@cwi.nl.}}\\
\and 
Haotian Jiang\thanks{University of Washington, Seattle, USA. \texttt{jhtdavid@cs.washington.edu.}} \\
\and 
Victor Reis\thanks{University of Washington, Seattle, USA. \texttt{voreis@cs.washington.edu. }}}
\date{}

\begin{document}

  


\maketitle
\pagenumbering{roman}
\vspace{-0.2cm}
\begin{abstract}
Motivated by the Matrix Spencer conjecture, we study the problem of finding signed sums of matrices with a small matrix norm. A well-known strategy to obtain these signs is to prove, given matrices $A_1, \dots, A_n \in \setR^{m \times m}$, a Gaussian measure lower bound of $2^{-O(n)}$ for a scaling of the discrepancy body  $\{x \in \mathbb{R}^n: \| \sum_{i=1}^n x_i A_i\| \leq 1\}$. We show this is equivalent to covering its polar with $2^{O(n)}$ translates of the cube $\frac{1}{n} B^n_\infty$, and construct such a cover via mirror descent. As applications of our framework, we show:

\medskip
\noindent \textbf{Matrix Spencer for Low-Rank Matrices.} If the matrices satisfy $\|A_i\|_\op \leq 1$ and $\rank(A_i) \leq r$, we can efficiently find a coloring $x \in \{\pm 1\}^n$ with discrepancy $\|\sum_{i=1}^n x_i A_i \|_\op \lesssim \sqrt{n \log (\min(rm/n, r))}$. This improves upon the naive $O(\sqrt{n \log r})$ bound for random coloring and proves the matrix Spencer conjecture when $r m \leq n$. 

\medskip
\noindent \textbf{Matrix Spencer for Block Diagonal Matrices.} For block diagonal matrices with $\|A_i\|_\op \leq 1$ and block size $h$, we can efficiently find a coloring $x \in \{\pm 1\}^n$ with $\|\sum_{i=1}^n x_i A_i \|_\op \lesssim \sqrt{n \log (hm/n)}$. This bound was previously shown in [Levy, Ramadas and Rothvoss, IPCO 2017] under the assumption $h \leq \sqrt{n}$, which we remove. Using our proof, we reduce the matrix Spencer conjecture to the existence of a $O(\log(m/n))$ quantum relative entropy net on the spectraplex. 

\medskip
\noindent \textbf{Matrix Discrepancy for Schatten Norms.} We generalize our discrepancy bound for matrix Spencer to Schatten norms $2 \le p \leq q$. Given $\|A_i\|_{S_p} \leq 1$ and $\rank(A_i) \leq r$, we can efficiently find a partial coloring $x \in [-1,1]^n$ with $|\{i : |x_i| = 1\}| \ge n/2$ and $\|\sum_{i=1}^n x_i A_i\|_{S_q} \lesssim \sqrt{n \min(p, \log(rk))} \cdot k^{1/p-1/q}$, where $k := \min(1,m/n)$.  

\medskip
\noindent Our partial coloring bound is tight when $m = \Theta(\sqrt{n})$. We also provide tight lower bounds of $\Omega(\sqrt{n})$ for rank-$1$ matrix Spencer when $m = n$, and $\Omega(\sqrt{\min(m,n)})$ for $S_2 \rightarrow S_\infty$ discrepancy, precluding a matrix version of the Koml\'os conjecture.

\end{abstract}

\clearpage
\setcounter{tocdepth}{2}

\tableofcontents
   
\newpage

\pagenumbering{arabic}
\setcounter{page}{1}

\section{Introduction}

Discrepancy minimization has been a well-studied area of research both in mathematics and computer science \cite{c01book,m09book}. 
We start with a classical setting: given vectors $a_1, \dots, a_n \in \mathbb{R}^m$ each satisfying $\|a_i\|_\infty \leq 1$, the goal is to find a coloring $x \in \{\pm 1 \}^n$ that minimizes the discrepancy, defined as $\| \sum_{i=1}^n x_i a_i \|_\infty$.
A seminal result of Spencer \cite{s85} improves upon the $O(\sqrt{n \log m})$ bound of a random coloring via Chernoff and union bound:

\begin{restatable}[Spencer \cite{s85}]{theorem}{spencer} \label{thm:spencer}
Let $m \geq n$. Given vectors $a_1, \dots, a_n \in \mathbb{R}^m$ with $\|a_i\|_\infty \leq 1$, there exists $x \in \{\pm 1\}^n$ such that $\| \sum_{i=1}^n x_i a_i \|_\infty \lesssim \sqrt{n \log(2m/n)}$. 
\end{restatable}

In particular, when $m = n$, \Cref{thm:spencer} states that the discrepancy is at most $O(\sqrt{n})$, as opposed to the $O(\sqrt{n \log n})$ bound for a random coloring. 
Spencer's theorem is known to be tight up to constants for all $m \ge n$ \cite{c01book,m09book}.

\medskip 
\noindent \textbf{The Partial Coloring Method.}  
All known proofs of Spencer's theorem are essentially based on the {\em partial coloring} method, one of the most important and widely applied techniques in discrepancy theory.
The method states that to obtain the type of discrepancy bound in \Cref{thm:spencer}, it suffices to prove the same bound for a partial coloring $x \in [-1,1]^n$ with at least $\Omega(n)$ coordinates in $\{\pm 1\}$. 
This process is then iterated over the set of coordinates $\{i: |x_i| < 1\}$ to obtain a full coloring. 
For Spencer-type problems, the discrepancy of the full coloring is at most a constant factor off from the discrepancy of the partial coloring (see \Cref{cor:full_coloring}).

The partial coloring method was developed in the early 80s by Beck and refined by Spencer using the entropy method \cite{b81b,s85}. A convex geometry view of partial coloring was developed independently by Gluskin \cite{g89}.
While these original arguments used the pigeonhole principle and were non-algorithmic, a breakthrough result of Bansal \cite{b10}, followed by a rich line of work \cite{lm15,r17,lrr17,es18,rr20}, gave various algorithmic versions. 
These recent developments also led to new results in approximation algorithms and differential privacy \cite{r13,ntz13,bclk14,bn17}.

\medskip
\noindent \textbf{Matrix Spencer Setting.} 
A natural generalization of Spencer's setting to matrices is the following. Given matrices $A_1, \dots, A_n \in \mathbb{R}^{m \times m}$, each satisfying $\| A_i \|_\op \leq 1$, the goal is to find a coloring $x \in \{\pm 1 \}^n$ that minimizes $\| \sum_{i=1}^n x_i A_i \|_{\op}$. 
In particular, Spencer's setting corresponds to the case where all matrices $A_i$ are diagonal. 

In the matrix Spencer setting, the non-commutative Khintchine inequality of Lust-Piquard and Pisier \cite{lpp91,p03book} shows that a random coloring $x \in \{\pm 1\}^n$ has expected discrepancy $\mathbb{E} [\|\sum_{i=1}^n x_i A_i \|_\op] \lesssim \sqrt{n \log r}$, where each matrix $A_i$ has rank at most $r \leq m$. 
It is conjectured that the discrepancy bound in \Cref{thm:spencer} can be generalized as follows:

\begin{conjecture}[Matrix Spencer Conjecture \cite{m14,z12}] \label{conj:matrix_spencer}
Let $m \ge \sqrt{n}$. Given matrices $A_1, \dots, A_n \in \mathbb{R}^{m \times m}$ with $\| A_i \|_\op \leq 1$,  there exists $x \in \{\pm 1\}^n$ such that 
\begin{align*}
\Big\| \sum_{i=1}^n x_i A_i \Big\|_\op \lesssim \sqrt{n \cdot \max(1, \log (m/n))}. 
\end{align*}
\end{conjecture}

In particular, when $\sqrt{n} \leq m \leq n$, the conjectured discrepancy bound is $O(\sqrt{n})$. 
Despite significant effort, \Cref{conj:matrix_spencer} has remained largely open, with partial progress for
block diagonal matrices \cite{lrr17}
and rank-$1$ matrices \cite{mss15,kls20}. A solution to \Cref{conj:matrix_spencer} will thus likely lead to new techniques and insights in discrepancy theory beyond what is currently known for vector discrepancy.

We note that in Spencer's setting (\Cref{thm:spencer}) we may assume without loss of generality that $m \geq n$ by the iterated rounding technique \cite{bf81,b08,lrs11}.
For matrix Spencer, however, the interesting regime starts at $m \geq \sqrt{n}$ (iterated rounding only works when $m^2 < n$). \Cref{conj:matrix_spencer} remains open even when $m = n^{1/2+ \varepsilon}$ for any constant $\varepsilon > 0$.

\medskip
\noindent \textbf{Matrix Discrepancy for Schatten Norms.} More generally, let\footnote{We make the assumption that $p \leq q$ to avoid a polynomial dependence on $m$ in the discrepancy bound. If $q < p$, then even a single matrix (i.e. $n=1$) can have discrepancy $m^{1/q - 1/p}$.} $2 \leq p \leq q \leq \infty$, we consider the following matrix discrepancy setting for Schatten norms. Given matrices $A_1, \dots, A_n \in \mathbb{R}^{m \times m}$, each satisfying $\| A_i \|_{S_p} \leq 1$, where $\| \cdot\|_{S_p}$ denotes the Schatten-$p$ norm. 
The goal is to find a coloring $x \in \{\pm 1\}^n$ to minimize $\| \sum_{i=1}^n x_i A_i \|_{S_q}$, the $S_p \rightarrow S_q$ discrepancy. 
In particular, the matrix Spencer setting corresponds to the case where $p = q = \infty$. 

The diagonal case of $S_p \rightarrow S_q$ discrepancy, i.e. $\ell_p \rightarrow \ell_q$ discrepancy for vectors, is well studied (see \cite{dntt18,rr20} and the references therein). In fact, the well-known Koml\'os conjecture asserts that the $\ell_2 \rightarrow \ell_\infty$ discrepancy can be upper bounded by a universal constant. 
For general $\ell_p \rightarrow \ell_q$ discrepancy, Reis and Rothvoss \cite{rr20} proves an optimal partial coloring bound of $O(\sqrt{\min(p, \log(m/n))} \cdot n^{1/2-1/p+1/q})$, assuming $m \geq n$ and $2 \leq p \leq q \leq \infty$. 
It is a natural question whether these bounds generalize to $S_p \rightarrow S_q$ discrepancy.

\medskip 
\noindent \textbf{The Challenge in Using Partial Coloring Method for Matrix Discrepancy.} 
Central to the partial coloring method is to show that the discrepancy body $D:=\{x \in \mathbb{R}^n: \| \sum_{i=1}^n x_i A_i \| \leq t\}$, i.e. the set of fractional colorings with discrepancy at most $t$ under norm $\|\cdot\|$, is ``large'' in some sense. 
A natural notion of largeness, due to Gluskin \cite{g89}, is that the body $D$ has Gaussian measure at least $2^{-O(n)}$. 
This measure of largeness has been adopted (sometimes implicitly) in essentially all work on partial coloring \cite{b10,lm15,r17,lrr17,es18,rr20}.

For the setting in \Cref{thm:spencer}, the discrepancy body $D$ is a polytope defined by the intersection of strips of the form $|\langle r_i, x \rangle| \leq t$, where $r_i \in \mathbb{R}^n$ are the rows of the $m \times n$ matix whose columns are $a_1, \dots, a_n$. 
Therefore, {\v{S}}id{\'a}k's lemma \cite{s67} can be readily used to give a Gaussian measure lower bound of the form $\gamma_n(D) \ge \prod_{i=1}^m \gamma_n (\{x \in \setR^n: |\langle r_i, x \rangle| \leq t\})$.

In the setting of matrix discrepancy, however, the discrepancy body $D$ has an infinite number of facets. This prevents the use of Gaussian correlation inequalities to lower bound $\gamma_n(D)$. To get around this barrier and use the partial coloring method for matrix discrepancy, one needs a different approach for proving Gaussian measure lower bounds.

\subsection{Our Results}

We lower bound the Gaussian measure of the discrepancy body $D$ via covering numbers for its polar $D^\circ$ with respect to the $\ell_\infty$-ball (see \Cref{sec:partial_coloring_via_covering}). 
We then prove the desired covering number estimates using mirror descent, the powerful convex optimization primitive of Nemirovski and Yudin \cite{ny83} (see \Cref{sec:mirror_descent,sec:covering_mirror,sec:entropy_net}). 
Our method yields the following applications.

\medskip
\noindent \textbf{Matrix Spencer for Low-Rank Matrices}. Our first result is the following improvement over the $O(\sqrt{n \log r})$ bound for random coloring in the matrix Spencer setting. 

\begin{restatable}[Matrix Spencer for Low-Rank Matrices]{theorem}{matrixspencer}\label{thm:weak_matrix_spencer}
Let $m \geq \sqrt{n}$. Given symmetric matrices $A_1, \dots, A_n \in \mathbb{R}^{m \times m}$ with $\|A_i\|_\op \leq 1$ and $\rank(A_i) \leq r$ for all $i \in [n]$, one can efficiently find a coloring $x \in \{\pm 1 \}^n$ such that 
\begin{align*}
\Big\|\sum_{i=1}^n x_i A_i \Big\|_\op \lesssim \sqrt{n \cdot \max(1,\log (r \cdot \min(1, m/n))) } .
\end{align*}
\end{restatable}

When the input matrices have rank $r \lesssim n/m$, the discrepancy bound in \Cref{thm:weak_matrix_spencer} is $O(\sqrt{n})$ and this proves \Cref{conj:matrix_spencer} for low rank matrices in the regime where $m \le n$.

\medskip
\noindent \textbf{Matrix Spencer for Block Diagonal Matrices.} Our second application is the following improved matrix Spencer bound for block diagonal matrices.

\begin{restatable}[Matrix Spencer for Block Diagonal Matrices]{theorem}{blockdiagonal} \label{thm:block_diagonal_matrix_spencer}
Let $m \ge \sqrt{n}$ and $h \leq m$. Given block diagonal symmetric matrices $A_1, \dots, A_n \in \mathbb{R}^{m \times m}$ with $\|A_i\|_\op \leq 1$ and block size $h \times h$, one can efficiently find a coloring $x \in \{\pm 1\}^n$ with 
\begin{align*}
\Big \|\sum_{i=1}^n x_i A_i \Big\|_\op \lesssim \sqrt{n \cdot \max(1, \log ( hm/n))}.
\end{align*}
\end{restatable}
 
In particular, \Cref{thm:block_diagonal_matrix_spencer} proves \Cref{conj:matrix_spencer} whenever $h \lesssim n/m$. 
This bound was previously proved in \cite{lrr17} under the assumption $h \leq \sqrt{n}$, which we remove here.

We also obtain the following reduction of \Cref{conj:matrix_spencer} to the construction of a better quantum relative entropy net for the spectraplex $\mathcal{S}_m := \{X \in \mathbb{R}^{m \times m}: X \succeq 0, \tr(X) = 1\}$. 

\begin{restatable}[Better Entropy Net Implies Matrix Spencer]{corollary}{betterentropynet} \label{cor:better_entropy_net}
Let $m \geq \sqrt{n}$. 
If we can find $T \subseteq \mathcal{S}_m$ with $|T| \leq 2^{O(n)}$ such that for each $X \in \mathcal{S}_m$ there exists $Y \in T$ with $S(X \| Y) \lesssim \max(1,\log (m/n))$, where $S(X \| Y)$ is the quantum relative entropy between $X$ and $Y$, then the matrix Spencer conjecture is true. 
\end{restatable}

In particular, in the proof of \Cref{thm:block_diagonal_matrix_spencer}, we construct a $O(\max(1,\log (hm/n)))$-relative entropy net for the set of block diagonal matrices on $\mathcal{S}_m$ with block size $h \times h$ (see \Cref{sec:entropy_net}). 
Our construction of such relative entropy nets might be of independent interest.

\medskip
\noindent \textbf{Matrix Discrepancy for Schatten Norms.} \Cref{thm:weak_matrix_spencer} is a special case of the following general matrix discrepancy bound for Schatten norms. 

\begin{restatable}[Matrix Discrepancy for Schatten Norms]{theorem}{matrixdisc}\label{thm:lowrank_matrix_discrepancy}
Let $m \geq \sqrt{n}$ and $2 \le p \le q \le \infty$. Given symmetric matrices $A_1, \dots, A_n \in \mathbb{R}^{m \times m}$ with $\|A_i\|_{S_p} \leq 1$ and $\mathrm{rank}(A_i) \le r$ for all $i \in [n]$, one can efficiently find $x \in [-1,1]^n$ so that $|\{i: |x_i| = 1\}| \ge n/2$ and 
\begin{align*}
\Big\|\sum_{i=1}^n x_i A_i \Big\|_{S_q} \lesssim \sqrt{n \cdot \min(p, \max(1,\log (rk)) )} \cdot k^{1/p-1/q},
\end{align*}
where we denote $k := \min(1, m/n)$. Moreover, we can find a full coloring $x \in \{\pm 1\}^n$ at the expense of a factor of $(1/2 + 1/q - 1/p)^{-1}$. 
\end{restatable}

Our partial coloring result in \Cref{thm:lowrank_matrix_discrepancy} is tight when either $m = \Theta(\sqrt{n})$ (for which we give an alternative proof using Banaszczyk's result \cite{b98} in \Cref{sec:banaszczyk_matrix}), or when $r=1$ and $m \geq n$. 
We provide matching lower bounds for both cases in \Cref{sec:lower_bound_low_dimension,sec:lower_bound_rank_1}. 
In particular, our lower bound examples imply a tight $\Omega(\sqrt{n})$ lower bound for rank-$1$ matrix Spencer when $m = n$.   

\begin{restatable}[Rank-$1$ Matrix Spencer Lower Bound]{corollary}{RankOneMatrixSpencer}
There exist rank-$1$ symmetric matrices $A_1, \dots, A_n \in \mathbb{R}^{n \times n}$ with $\|A_i\|_\op \leq 1$ such that any $x \in \{\pm 1\}^n$ has $\|\sum_{i=1}^n x_i A_i\|_\op \gtrsim \sqrt{n}$.
\end{restatable}

Another immediate consequence of our lower bounds is an optimal $\Omega(\sqrt{\min(m,n)})$ lower bound for $S_2 \rightarrow S_\infty$ discrepancy. 
This is in stark contrast to the well-known Koml\'os conjecture for vectors, which asserts that the $\ell_2 \rightarrow \ell_\infty$ discrepancy is $O(1)$. 
\Cref{komlos_lower} states that such a conjecture is far from being true for matrices. 

\begin{restatable}[Lower Bound for Matrix Koml\'os]{corollary}{MatrixKomlos} \label{komlos_lower}
For any $m$ and $n$, there exist symmetric matrices $A_1, \dots, A_n \in \mathbb{R}^{m \times m}$ with $\|A_i\|_F \leq 1$ such that any $x \in \{\pm 1\}^n$ has $\|\sum_{i=1}^n x_i A_i\|_\op \gtrsim \sqrt{\min(m,n)}$. 
\end{restatable}

Finally, we propose the following generalization of \Cref{conj:matrix_spencer}:

\begin{conjecture}[$S_p \to S_q$ Matrix Discrepancy]
Let $m \ge \sqrt{n}$  and $2 \le p \le q \le \infty$. Given matrices $A_1, \dots, A_n \in \mathbb{R}^{m \times m}$ with $\| A_i \|_{S_p} \leq 1$,  there exists $x \in \{\pm 1\}^n$ such that 
\begin{align*}
\Big\| \sum_{i=1}^n x_i A_i \Big\|_{S_q} \lesssim \sqrt{n \cdot \min(p,\max(1, \log (m/n)))} \cdot \min(1,m/n)^{1/p-1/q}. 
\end{align*}
\end{conjecture}

When $m = n$, the right hand side is $O(\sqrt{n})$, and for diagonal matrices the conjecture is known to be true for any $2 \le p \le q$. When $p = q$, the conjecture is also known to be true for diagonal matrices for all $m$ and $n$ \cite{rr20}.

\subsection{Overview of Our Approach}

We give a brief overview of our partial coloring framework in this subsection, and leave a more detailed discussion to \Cref{sec:framework}.

\medskip 
\noindent \textbf{Partial Coloring via Covering Numbers.}
Let $K := \{x \in \mathbb{R}^n: \| \sum_{i=1}^n x_i A_i \| \leq 1\}$ be the unit discrepancy body\footnote{To avoid confusion when talking about discrepancy bodies, $K$ denotes the unit discrepancy body, and $D$ denotes a scaling of $K$ by the target discrepancy bound.} and $t$ be the target discrepancy bound.  
A recent refinement by Reis and Rothvoss \cite{rr20} of Gluskin's convex geometry approach \cite{g89} shows that whenever $\gamma_n(tK) \geq 2^{-O(n)}$ for any constant in the exponent, one can efficiently find a partial coloring $x \in O(t K) \cap [-1,1]^n$ with at least $n/2$ coordinates in $\{-1,1\}$ (see \Cref{thm:partial_coloring}). For settings where the target discrepancy bound is $n^{\Omega(1)}$, we may iterate the partial coloring to find a full coloring with the same discrepancy bound up to constants (see \Cref{cor:full_coloring}). 

Our new approach for proving a Gaussian measure lower bound $\gamma_n(tK) \geq 2^{-O(n)}$ is via the covering numbers (\Cref{defn:covering_numbers}) of $K$ or $K^\circ$ with respect to the Euclidean ball $B_2^n$ or the $\ell_\infty$ ball $B_\infty^n$. 
In particular, since $\gamma_n(\sqrt{n} B_2^n)$ has constant Gaussian measure, as long as $\mathcal{N}(\sqrt{n} B_2^n, t K) \leq 2^{O(n)}$, we get $\gamma_n(tK) \geq 2^{-O(n)}$.
Using the duality of covering numbers and connections with volume, we obtain several equivalent conditions for $\gamma_n(tK) \geq 2^{-O(n)}$ in terms of covering (\Cref{lem:dual_cov}). 
The condition that we will work with is $\mathcal{N}(K^\circ, \frac{t}{n} B_\infty^n) \leq 2^{O(n)}$, where $K^\circ = \{(\langle A_1, U \rangle, \dots, \langle A_n, U \rangle): \|U\|_* \leq 1\}$ is the polar discrepancy body.

\medskip
\noindent \textbf{Covering via Mirror Descent.} 
We prove the covering number bound $\mathcal{N}(K^\circ, \frac{t}{n} B_\infty^n) \leq 2^{O(n)}$ using mirror descent, a powerful convex optimization primitive of Nemirovski and Yudin \cite{ny83} (see \Cref{sec:mirror_descent} for an overview). 
In particular, denote the linear map $\mathcal{A}(U) := (\langle A_1, U \rangle, \dots, \langle A_n, U \rangle)$.
We shall assume that each $\|A_i\| \leq 1$. 
This is true for the matrix Spencer setting with $\| \cdot\|$ being the operator norm. 
In the more general setting of matrix discrepancy for Schatten norms, we have $\|A_i\|_{S_p} \leq 1$ while the norm for measuring discrepancy is $\|\cdot\|_{S_q}$. 
One can get around this issue by leveraging known covering number estimates between Schatten classes (\Cref{thm:entropy_number_Schatten}).

For any matrix $\|U \|_* \leq 1$, consider minimizing the function $f_U(X) := \|\mathcal{A}(X - U)\|_\infty$ over the dual unit ball $B_* := \{U: \|U\|_* \leq 1\}$.
The function has minimum value $f_U(U) = 0$ and since it has subgradients in $\{\pm A_1, \dots,\pm A_n\}$ with $\|A_i\| \leq 1$, the function $f_U(X)$ is $1$-Lipschitz with respect to the dual norm $\| \cdot\|_*$. 
So as long as there exists a $1$-strongly convex mirror map $\Phi$ on $B_*$, we can minimize $f_U(X)$ by starting from some matrix $U_0 = U_0(U) \in B_*$ and running mirror descent for $n$ steps. Denoting by $U_s$ the matrix in the $s$-th step, standard guarantees for mirror descent (\Cref{thm:mirror_descent}) yield
\begin{align} \label{eq:error_mirror_descent}
\min_{s \in [n]} f_U(U_s) = \min_{s \in [n]} f_U(U_s) - f_U(U) \leq \sqrt{\frac{2 D_\Phi(U, U_0)}{n}} ,
\end{align}
where $D_\Phi(U,U_0) = \Phi(U) - \Phi(U_0) - \langle \nabla \Phi(U_0), U - U_0 \rangle$ is the Bregman divergence. 
We let $T$ be the set of all matrices encountered when running mirror descent for all possible $U \in B_*$, i.e.
$T := \{U_s: s \in [n], U \in B_*\}$, and $T_0 := \{U_0: U \in B_*\}$ be the set of all starting matrices. 
The net $\mathcal{A}(T)$ will be our covering for $K^\circ$. 

To see that this indeed gives a good covering, we denote $D_\Phi^{\max} := \sup_{U \in B_*} D_\Phi(U \| U_0)$. 
By the definition of the function $f_U$, we have from \eqref{eq:error_mirror_descent} that
\begin{align*}
\min_{s \in [n]} \| \mathcal{A}(U) - \mathcal{A}(U_s) \|_\infty \leq \sqrt{\frac{2 D_\Phi(U, U_0)}{n}} \leq \sqrt{\frac{2 D_\Phi^{\max}}{n}},
\end{align*} 
and so the dual body admits the covering $K^\circ \subseteq \mathcal{A}(T) + \sqrt{2 D_\Phi^{\max} / n} \cdot B_\infty^n$. 
Thus as long as our target discrepancy bound $t \le \sqrt{2 n D_\Phi^{\max}}$, we have $\mathcal{N}(K^\circ, \frac{t}{n} B_\infty^n) \leq |T|$, which we need to show to be at most $2^{O(n)}$. 

The key observation we make here is that for our choices of the mirror maps in \Cref{sec:weak_matrix_spencer,sec:matrix_discrepancy_schatten}, $U_s$ only depends\footnote{In general, mirror descent projects back onto the feasible set according to the Bregman divergence in each iteration, and therefore might not satisfy this property.} on the sum of the subgradients, but not on their order. 
Since there are only $2n$ choices of subgradients $\{\pm A_i\}_{i \in [n]}$ and we run mirror descent for $n$ steps, a counting argument reveals that there are at most $2^{O(n)}$ possible sums of gradients (\Cref{lem:size_covering}). 
So long as the starting matrices satisfy $|T_0| \leq 2^{O(n)}$, we have $|T| \leq |T_0| \cdot 2^{O(n)} \leq 2^{O(n)}$.

\medskip
\noindent \textbf{A View of Mirror Descent as Refining the Net.}
In the diagonal case, i.e. Spencer's setting, we can directly build the net $T$ by repeatedly sampling the $i$th diagonal coordinate $e_i e_i^\top$ proportional to its weight in the target matrix. 
Since the set of diagonal matrices on the Schatten-$1$ ball has only $2m$ vertices $\{\pm e_i e_i^\top\}_{i \in [m]}$, the approximate Carath\'eodory theorem (see \cite{v18book}, Theorem 0.0.2) implies that the image of the net $\mathcal{A}(T)$ already gives a good covering for $K^\circ$, and mirror descent is not necessary in this case.

However, this argument fails beyond diagonal matrices, as the number of vertices becomes infinite. In these more general cases, we use mirror descent to boost a coarse net $T_0$ to a finer net $T$ which has a better covering guarantee in the image space, at the expense of increasing the size of the net by a factor of $2^{O(n)}$.

\medskip
\noindent \textbf{Relative Entropy Nets for the Spectraplex.} 
For our application in \Cref{sec:matrix_discrepancy_schatten} to low-rank matrices, it suffices to take $T_0 = \{0\}$. For the application in \Cref{sec:weak_matrix_spencer} to block diagonal matrix Spencer, we run mirror descent on the spectraplex $\mathcal{S}_m := \{X \in \mathbb{R}^{m \times m}: X \succeq 0, \tr(X) = 1\}$ and carefully construct a set $|T_0| \leq 2^{O(n)}$ with small $D_\Phi^{\max}$.  
Since $D_\Phi(X \| Y)$ is the quantum relative entropy between $X$ and $Y$ in the spectraplex setup, we refer to such $T_0$ as a (quantum) relative entropy net (\Cref{defn:relative_entropy_net}).

We use an operator norm net for the Schatten-1 ball from \cite{hpv17} to construct a relative entropy net with error $O(\log(m^2/n))$ for the spectraplex $\mathcal{S}_m$ (\Cref{lem:entropy_from_op}). 
When restricted to block diagonal matrices with block size $h \times h$, we use a hybrid of this argument and the earlier approximate Caratheodory argument to find a refined relative entropy net with error $O(\log(hm/n))$ (\Cref{thm:entropy_net}). Taking $T_0$ to be this net in our mirror descent framework gives \Cref{thm:block_diagonal_matrix_spencer}.
This also allows us to reduce the matrix Spencer conjecture to the existence of a better relative entropy net with error $O(\log (m/n))$ for the spectraplex (\Cref{cor:better_entropy_net}).

\subsection{Further Related Work}

\medskip
\noindent \textbf{Banaszczyk's Approach.} 
While the partial coloring method has been extensively applied in discrepancy and obtains the optimal bound for many problems, for several applications where the target discrepancy bound is $n^{o(1)}$ (e.g. the Koml\'os problem or Tusnady's problem), partial coloring is potentially sub-optimal by a logarithmic factor. In breakthrough work, Banaszczyk \cite{b98} obtained an improvement over the partial coloring method for these applications using deep techniques from convex geometry. 
While Banaszczyk's original proof is non-constructive, a fascinating recent line of work has obtained algorithmic versions of Banaszczyk's result \cite{dgln16,bdg16,bg17,lrr17,bdgl18}.

\medskip
\noindent \textbf{Matrix Spencer Conjecture and Non-commutative Random Matrix Theory.}
The typical value of $\|\sum_{i=1}^n x_i A_i\|_\op$ for a random coloring has attracted significant attention in random matrix theory. 
For commutative matrices, the bound $\E[\|\sum_{i=1}^n x_i A_i\|_\op] \lesssim \sqrt{n \log m}$ by matrix Khintchine \cite{lpp91,p03book} or matrix Chernoff bound \cite{aw02} is in general tight. It is also known to be tight for Toeplitz matrices \cite{m07}. For matrices with certain non-commutative structures (e.g. random Gaussian matrices), improved bounds of $O(\sqrt{n})$ are known (see \cite{v18book,bbvh21}). 
In the context of \Cref{conj:matrix_spencer}, these results imply that a random coloring already achieves the conjectured bound when the input matrices have certain non-commutative structures. 
On the other hand, by \Cref{thm:spencer}, \Cref{conj:matrix_spencer} is known when all the matrices commute.

\medskip \noindent \textbf{Concurrent and Independent Work.}
In concurrent and independent work, Hopkins, Raghavendra and Shetty \cite{hrs21} proved a bound of $\sqrt{n \log \left(\tr(\sum_{i=1}^n A_i^2 ) / n^{1.5} \right)}$ 
for matrix Spencer using quantum communication complexity. 
Their bound coincides with ours for full rank matrices, and is slightly stronger for low-rank matrices. 
However, our approach is completely different and can also be used to show matrix discrepancy bounds for block diagonal matrices and general Schatten norms.
We believe both approaches are interesting and may lead to further progress in resolving the matrix Spencer conjecture.

\section{Preliminaries}

\paragraph{Norms and Convex Bodies.} 
A convex body is a compact convex set with non-empty interior. We say a convex set $K$ is symmetric if $x \in K$ implies $-x \in K$. 
We use $\|\cdot\|_p$ to denote the $\ell_p$-norm and $\|\cdot \|_{S_p}$ to denote the Schatten-$p$ norm.
In particular, the operator norm $\|\cdot\|_\op = \|\cdot \|_{S_\infty}$ and the Frobenius norm $\|\cdot\|_F = \|\cdot \|_{S_2}$.
We use $B_p^n$ to denote the unit $\ell_p$-ball in $\mathbb{R}^n$ and $B_{S_p}^n := \{A \in \setR^{n \times n}: \|A\|_{S_p} \le 1\}$ 
to denote the unit Schatten-$p$ ball in $\mathbb{R}^{n \times n}$, with $B_{\mathrm{op}}^n := B_{S_\infty}^n$. 
Let $\mathbb{R}^n_+$ denote the set of non-negative vectors in $\mathbb{R}^n$ and denote the simplex $\Delta_n := \{x \in \mathbb{R}^n_+: \|x\|_1 = 1\}$. 
Let $\mathbb{S}_+^n$ (resp. $\mathbb{S}_{++}^n$) denote the set of positive semidefinite (resp. positive definite) $n \times n$ matrices, and define the spectraplex $\mathcal{S}_n := \{X \in \mathbb{S}_+^n: \tr(X) = 1\}$. 
For a norm $\| \cdot \|$ in $\mathbb{R}^n$, we define the dual norm as $\|x\|_* := \sup\{\langle y, x \rangle: y \in \mathbb{R}^n, \|y\|\leq 1\}$. 
Dual norms are similarly defined for matrix norms.

\medskip
\noindent \textbf{Convex Functions.} A convex function $f: \mathcal{X} \rightarrow \mathbb{R}$ is said to be $L$-Lipschitz with respect to a norm $\| \cdot \|$ if $\|g\|_* \le L$ for all subgradients $g \in \partial f(x)$. 
We say that $f$ is $\alpha$-strongly convex with respect to a norm $\|\cdot\|$ if $f(y) \geq f(x) + g^\top (y - x) + \frac{\alpha}{2}\|x-y\|^2$, for all $x, y \in \mathcal{X}$ and all subgradients $g \in \partial f(x)$.

\medskip
\noindent \textbf{Polar.} Given a convex set $K \subseteq \mathbb{R}^n$ with $0 \in K$, we define the polar of $K$ to be $K^\circ := \{y \in \mathbb{R}^n: \sup_{x \in K} \langle x, y \rangle \leq 1\}$. 
It is immediate from the definition that for any constant $t > 0$, $(t K)^\circ = \frac{1}{t} K^\circ$.
When $K$ is closed, the polarity theorem states that $(K^\circ)^{\circ} = K$.  

\begin{lemma}[Polar of Discrepancy Set] \label{lem:polar_body}
Given matrices $A_1, \dots, A_n \in \mathbb{R}^{m \times m}$ and a norm $\| \cdot\|$ in $\mathbb{R}^{m \times m}$, we define the unit discrepancy set as $K := \{x \in \mathbb{R}^n: \| \sum_{i=1}^n x_i A_i \| \leq 1 \}$. Then $K' := \{(\langle A_1, U \rangle, \dots, \langle A_n, U \rangle): \|U\|_* \leq 1\}$ is the polar body $K' = K^\circ$.
\end{lemma}

\begin{proof}
By the definition of polar body, we may write
\begin{align*}
(K')^\circ & = \Big\{x \in \setR^n : \sum_{i=1}^n x_i \inn{A_i}{U} \le 1, \ \forall \ U  \text{ s.t. } \|U\|_* \le 1\Big\} \\ & = \Big\{x \in \setR^n :  \Big\langle \sum_{i=1}^n x_i A_i, U \Big\rangle \le 1, \ \forall \ U \text{ s.t. } \|U\|_* \le 1\Big\} \\ & = K,
\end{align*} by the definition of dual norm. It then follows from the polarity theorem that $K' = K^\circ$. 
\end{proof}



\medskip
\noindent \textbf{Covering Numbers.} We start with the definition of covering numbers.

\begin{definition}[Covering Numbers] \label{defn:covering_numbers}
For two convex bodies $K,T \subseteq \mathbb{R}^n$, we define the covering number $\mathcal{N}(K, T)$ as the minimum number $N$ such that there exist centers $x_1, \dots, x_N \in \mathbb{R}^n$ with $K \subseteq \cup_{i=1}^N (x_i + T)$, i.e. $K$ can be covered by $N$ translates of $T$.
\end{definition}

We need the following few standard facts about covering numbers (see \cite{agm15}).

\begin{lemma}[Volume Bounds for Covering Numbers] \label{lem:covering_vol_bound}
Given convex bodies $K, T \subseteq \mathbb{R}^n$. If $T$ is symmetric, we have $\frac{\vol_n(K)}{\vol_n(T)} \leq \mathcal{N}(K,T) \leq 2^n \cdot \frac{\vol_n(K + \frac{T}{2})}{\vol_n(T)}$. 
\end{lemma}

\begin{lemma}[Symmetrization] \label{lem:symmetrization_covering}
Let $K \subseteq \mathbb{R}^n$ be a convex body, then $\mathcal{N}(K-K, K) \leq 2^{O(n)}$. 
\end{lemma}

\begin{theorem}[Duality of Covering Numbers, \cite{km87}] \label{thm:cov_duality}
Given symmetric convex bodies $K, T \subseteq \mathbb{R}^n$, we have
\begin{align*}
2^{-\Theta(n)} \cdot \mathcal{N}(T^\circ, K^\circ) \leq \mathcal{N}(K,T) \leq 2^{\Theta(n)} \cdot \mathcal{N}(T^\circ, K^\circ) .
\end{align*}
\end{theorem}

We will also need the following upper bound on the covering numbers of Schatten balls\footnote{We note that \cite{hpv17} claims the bound only up to a constant depending on $p$ and $q$, but their argument readily gives a universal constant in the regime of $p, q \ge 1$.}.

\begin{theorem}[\cite{hpv17}, Theorem 1.1] \label{thm:entropy_number_Schatten} Let $m, n \in \N$ and $1 \le p \le q \le \infty$. Then we have 
\begin{align*}
\mathcal{N}\Big(B_{S_p}^m, \min\Big(1,\frac{m}{n}\Big)^{1/p - 1/q} B_{S_q}^m\Big) \le 2^{O(n)}.
\end{align*}
\end{theorem}

\medskip
\noindent \textbf{Gaussian Measure.} We use $\gamma_n(\cdot)$ to denote the standard Gaussian measure on $\mathbb{R}^n$. Gaussian measure is log-concave, i.e. $\gamma_n(\lambda A + (1-\lambda) B ) \geq \gamma_n(A)^{\lambda} \gamma_n(B)^{1 - \lambda}$ for any compact subsets $A, B \subseteq \mathbb{R}^n$. In particular, by taking $A = -x + K$ and $B = x + K$ for any $x \in \mathbb{R}^n$ and symmetric convex body $K$, and $\lambda = 1/2$, we have the following lemma.

\begin{lemma}[Translation Decreases Gaussian Measure] \label{lem:translate_gauss}
Given any symmetric convex body $K \subseteq \mathbb{R}^n$ and $x \in \mathbb{R}^n$, we have $\gamma_n(K) \geq \gamma_n(x + K)$.  
\end{lemma}

We also use the following powerful Gaussian correlation inequality. 

\begin{theorem}[Gaussian Correlation Inequality, \cite{r14gaussian}] \label{thm:Gaussian_correlation}
Given any symmetric convex sets $K,T \subseteq \mathbb{R}^n$, we have $\gamma_n(K \cap T) \geq \gamma_n(K) \cdot \gamma_n(T)$. 
\end{theorem}

\section{Our Framework for Partial Coloring}
\label{sec:framework}

\subsection{Partial Coloring via Covering Numbers}
\label{sec:partial_coloring_via_covering}

Given symmetric matrices $A_1, \dots, A_n \in \mathbb{R}^{m \times m}$, a norm $\| \cdot \|$ on $\mathbb{R}^{m \times m}$ for measuring the discrepancy, and a target discrepancy bound $t$, let $D := \{x \in \mathbb{R}^n: \|\sum_{i=1}^n x_i A_i \| \leq t\}$ be the associated discrepancy body. 
The following partial coloring lemma from \cite{rr20} states that one can efficiently find a partial coloring with discrepancy $O(t)$ as long as $\gamma_n(D) \geq 2^{-O(n)}$.

\begin{theorem}[\cite{rr20}, special case of Theorem 6] \label{thm:partial_coloring}
For any constant $\alpha > 0$, there is a constant $c:= c(\alpha) > 0$ and a randomized polynomial time algorithm that for a symmetric convex set $D \subseteq \mathbb{R}^n$  with $\gamma_n(D) \geq 2^{-\alpha n}$ and a shift $y \in (-1,1)^n$, finds $x \in (c\cdot D) \cap [-1,1]^n$ so that $x+y \in [-1,1]^n$ and $|\{i \in [n]: |(x+y)_i|=1\}| \geq n/2$. 
\end{theorem}

We have the following corollary for full colorings. Here $K_S := K \cap \{x \in \setR^n : x_i = 0, \forall i \notin S\}$.

\begin{corollary} \label{cor:full_coloring} Let $K \subseteq \mathbb{R}^n$ be a symmetric convex set. Given a function $f : [n] \to \setR_{> 0}$ with $\gamma_S (f(|S|) \cdot K_S) \ge 2^{-O(|S|)}$ for every $S \subseteq [n]$, there exists a randomized polynomial time algorithm to find a full coloring $x \in \{\pm 1\}^n$ so that $x \in \lambda K$, where $\lambda \lesssim \sum_{i=0}^{\lfloor \log n \rfloor} f(n/2^i)$. In particular, when $f(n) \lesssim n^{\beta}$ for some $\beta \le 1$, we have $\lambda \lesssim \frac{1}{\beta} n^\beta$.
\end{corollary}

\begin{proof}
Indeed, repeated iterations of \Cref{thm:partial_coloring} with $y_0 := 0$ and subsequent shifts $y_{i+1}$ being the coordinates not reaching $\{-1,1\}$ find $x := x_0 + \dots + x_T \in \{\pm 1\}^n$ for $T:= \lfloor \log n\rfloor$ with $x_t \in O(f(n/2^t)) \cdot K$. When $f(n) \lesssim n^{\beta}$, the summation is upper bounded by
\[
\sum_{i=0}^\infty (n/2^i)^\beta = (1-2^{-\beta})^{-1} \cdot n^\beta \lesssim  \frac{1}{\beta} \cdot n^\beta , 
\]
and this proves the statement. 
\end{proof}

We show that a $2^{-O(n)}$ Gaussian measure lower bound is equivalent to a $2^{O(n)}$ upper bound for certain covering numbers. 

\begin{lemma} \label{lem:dual_cov}
The following conditions are equivalent for a symmetric convex body $D \subseteq \setR^n$:
\begin{enumerate}
\item $\gamma_n(D) \geq 2^{-O(n)}$, 
\item $\mathcal{N}(\sqrt{n} B_2^n, D) \leq 2^{O(n)}$,
\item $\mathcal{N}(n B_1^n, D) \leq 2^{O(n)}$,
\item $\mathcal{N}(D^\circ, \frac{1}{\sqrt{n}} B_2^n) \leq 2^{O(n)}$, 
\item $\mathcal{N}(D^\circ, \frac{1}{n} B_\infty^n) \leq 2^{O(n)}$. 
\end{enumerate}
\end{lemma}

\begin{proof}
We start by proving that condition (1) implies (2). 
Suppose $\gamma_n(D) \geq 2^{-O(n)}$, then \Cref{thm:Gaussian_correlation} implies $\gamma_n(D') \geq 2^{-O(n)}$, where we define $D' := D \cap \sqrt{n} B_2^n$. 
We thus also have $\vol_n(D') \geq \gamma_n(D') \geq 2^{-O(n)}$. 
Then by \Cref{lem:covering_vol_bound}, we have
\begin{align*}
\mathcal{N}(\sqrt{n} B_2^n, D) \leq \mathcal{N}(\sqrt{n} B_2^n, D') & \leq 2^n \cdot \frac{\vol_n(\sqrt{n}B_2^n + D')}{\vol_n(D')} \leq 2^{n} \cdot \frac{\vol_n(2\sqrt{n}B_2^n)}{\vol_n(D')} \leq 2^{O(n)} .
\end{align*}
We next show that condition (2) implies (1). Since $\gamma_n(\sqrt{n} B_2^n) = \Omega(1)$, we have $\gamma_n(x + D) \geq 2^{-O(n)}$ for some $x \in \mathbb{R}^n$. 
\Cref{lem:translate_gauss} then gives $\gamma_n(D) \geq \gamma_n(x + D) \geq 2^{-O(n)}$.

The implication $(3) \Rightarrow (2)$ immediately follows from $\sqrt{n}B_2^n \subseteq n B_1^n$. 
To prove the reverse implication $(2) \Rightarrow (3)$, we use \Cref{lem:covering_vol_bound} to obtain 
\begin{align*}
\mathcal{N}(\sqrt{n} B_1^n, B_2^n) \leq 2^{n} \cdot \frac{\vol_n(\sqrt{n} B_1^n+ B_2^n)}{\vol_n(B_2^n)} \leq 2^{O(n)} \cdot \frac{\vol_n(\sqrt{n} B_1^n)}{\vol_n(B_2^n) } \leq 2^{O(n)}.
\end{align*} 
It thus follows that $\mathcal{N}(n B_1^n, D) \leq \mathcal{N}(n B_1^n, \sqrt{n} B_2^n) \cdot \mathcal{N}(\sqrt{n} B_2^n, D) \leq 2^{O(n)}$.

The last two equivalences follow from the duality of covering numbers in \Cref{thm:cov_duality}.
\end{proof}

For our mirror descent framework, we use the following corollary: 

\begin{corollary} \label{partial_covering}
Given matrices $A_1, \dots, A_n \in \setR^{m \times m}$, let $K^\circ_{q+} := \{(\inn{A_1}{U}, \dots, \inn{A_n}{U}) : U \in B^m_{S_q}, U \succeq 0\}$. If we have $\mathcal{N}(K^\circ_{q+}, \frac{t}{n} B^n_\infty) \leq 2^{O(n)}$, then we can efficiently find a partial coloring $x \in [-1,1]^n$ with $|\{i:|x_i| =1\}| \ge n/2$ and $\|\sum_{i=1}^n x_i A_i\|_{S_q} \lesssim t$.
\end{corollary}
{}
\begin{proof}
Recall that $D := \{x \in \mathbb{R}^n: \|\sum_{i=1}^n x_i A_i \|_{S_q} \leq t\}$ denotes the discrepancy body. Since $tD^\circ \subseteq K^\circ_{q+} - K^\circ_{q+}$, by \Cref{lem:symmetrization_covering} we have $\mathcal{N}(D^\circ, \frac{1}{n} B^n_\infty) = 2^{O(n)}$. The equivalence $(1) \Leftrightarrow (5)$ in \Cref{lem:dual_cov} implies $\gamma_n(D) \ge 2^{-O(n)}$, and \Cref{thm:partial_coloring} gives the corollary.  
\end{proof}

\subsection{Mirror Descent: An Overview} \label{sec:mirror_descent}

The mirror descent method was introduced by Nemirovski and Yudin \cite{ny83}. Here, we follow the presentation in \cite{b15}. 
Let $\mathcal{D}$ be an open subset of $\setR^m$ and $\mathcal{X}$ a subset of its closure. We fix a convex function $f : \mathcal{X} \to \mathbb{R}$ assumed to be $L$-Lipschitz with respect to a norm $\|\cdot \|$, and a differentiable function $\Phi : \mathcal{D} \to \setR$ that is $\rho$-strongly convex with respect to $\|\cdot\|$ and has a surjective gradient $\nabla \Phi : \mathcal{D} \to \setR^m$. The mirror descent algorithm, given a starting point $x_0 \in \mathcal{X} \cap \mathcal{D}$, consists of the iterations 
\begin{align*}
& \nabla \Phi(y_{t+1}) := \nabla \Phi(x_t) - \eta g_t, 
\\ & x_{t+1} := \mathrm{argmin}_{x \in \mathcal{X} \cap \mathcal{D}} D_\Phi (x, y_{t+1}),
\end{align*}
where $g_t \in \partial f(x_t)$ and $D_\Phi(x,y) := \Phi(x) - \Phi(y) - \nabla \Phi(y)^\top (x- y )$ is the Bregman divergence. Note that $y_t \in \mathcal{D}$ and $x_t \in \mathcal{X} \cap \mathcal{D}$ for all $t \ge 0$. We use the following convergence guarantee:

\begin{theorem}[\cite{b15}, Theorem 4.2] \label{thm:mirror_descent} Let $f$ be $L$-Lipschitz and $\Phi$ be $\rho$-strongly convex with respect to $\|\cdot \|$, and $D_\Phi^{\max} \geq D_\Phi(x^*, x_0)$ be any upper bound.
Then the mirror descent algorithm with $\eta := \frac{1}{L} \sqrt{\frac{2\rho D_\Phi^{\max}}{T}}$ satisfies
\begin{align*}
\min_{s \in [T]} f(x_s) - f(x^*) \leq L \sqrt{\frac{2 D_\Phi^{\max}}{\rho T}}.
\end{align*} 
\end{theorem}

\noindent \textbf{The Spectraplex Setup.} 
Here we take $\mathcal{X} := \mathcal{S}_m = \{X \in \mathbb{S}_+^m: \tr(X) = 1\}$. 
The mirror map is $\Phi(X) = \tr(X \log X)$, defined on $\mathcal{D} = \mathbb{S}_{++}^m$, which is $\frac{1}{2}$-strongly convex with respect to the Schatten-$1$ norm by the quantum Pinsker inequality \cite{c16}. 
Then the convergence bound in \Cref{thm:mirror_descent} becomes $2L \sqrt{\frac{S(X^* \| X_0)}{T}}$, where $S(X \| Y) := \tr(X(\log X - \log Y))$ is the quantum relative entropy between matrices $X, Y \in \mathcal{S}_m$. The projection step corresponds to a trace normalization, so given a starting point $X_0 \in \mathcal{S}_m \cap  \mathbb{S}_{++}^m$, we may write in closed form
\begin{align} \label{eq:spectraplex_iterations}
X_{t+1} = \frac{\exp(\log X_0 - \eta \sum_{i=0}^t g_i)}{\tr(\exp(\log X_0 - \eta \sum_{i=0}^t g_i))},
\end{align}
for subgradients $g_i \in \partial f(X_i)$.

\medskip
\textbf{The Schatten Norm Setup.}
Here we take $\mathcal{X} = \mathcal{D} = \setR^{m \times m}$, so that $X_t = Y_t$ for all $t$. The mirror map is $\Phi(X) := \frac{1}{2(p-1)} \|X\|_p^2$, which is known to be $1$-strongly convex for all $p \in (1, 2]$ \cite{bcl94}. Thus given a starting point $X_0 \in \setR^{m \times m}$, we may write in closed form 
\begin{align} 
\label{eq:schatten_iterations}
X_{t+1} = \nabla \Phi^{-1} \Big(\nabla \Phi(X_0) - \eta \sum_{i=0}^t g_i\Big),
\end{align} 
for subgradients $g_i \in \partial f(X_i)$.

\subsection{Covering via Mirror Descent} \label{sec:covering_mirror}
Given symmetric matrices $A_1, \dots, A_n$ with $\|A_i\| \le 1$ for all $i \in [n]$, where the dual norm $\| \cdot \|_*$ is either the Schatten-1 norm or the Schatten-$p$ norm for some $p \in (1,2]$, we apply mirror descent on functions of the form $\displaystyle f_U(X) := \max_{i \in [n]} |\inn{A_i}{X-U}|$ to cover the polar discrepancy body
\begin{align*} 
K^\circ := \{\mathcal{A}(U) : \|U\|_* \le 1\}, \text{ where } \mathcal{A}(U) := (\inn{A_1}{U} ,\dots, \inn{A_n}{U}).
\end{align*}
Note that $f_U(X) = \|\mathcal{A}(X) - \mathcal{A}(U)\|_\infty$ and that $f$ is $1$-Lipschitz with respect to $\| \cdot \|_*$. The key property of such functions is that we may always choose subgradients from the set of $2n$ matrices $\{\pm A_i : i \in [n]\}$, which allows us to upper bound the number of different matrices encountered during the mirror descent process. 

\begin{lemma} \label{lem:size_covering}
Let $\| \cdot\|_*$ be either $\|\cdot\|_{S_1}$ as in the Spectraplex Setup, or $\|\cdot\|_{S_p}$ with $p \in (1,2]$ as in the Schatten Norm Setup, and $\mathcal{X}, \mathcal{D}$ be defined accordingly.
Let $T_0 \subseteq \mathcal{X} \cap \mathcal{D}$ be a set with size $|T_0| \le 2^{O(n)}$ and $K^\circ \supseteq K' = \mathcal{A}(T')$ the convex body to be covered, where $T' \subseteq \mathcal{X} \cap \mathcal{D}$. If for every $U \in T'$ there exists a starting point $U_0 := U_0(U) \in T_0$ with $D_\Phi(U, U_0) \le D_\Phi^{\max}$, then we can bound 
\begin{align*}
\mathcal{N}\Big(K', \sqrt{\frac{D_\Phi^{\max}}{n}} B^n_\infty\Big) \le 2^{O(n)}.
\end{align*}
\end{lemma}

\begin{proof}
The key observation is that in either setup of mirror descent, the point $X_t$ in \eqref{eq:spectraplex_iterations} or \eqref{eq:schatten_iterations} depends only on the starting point $U_0$ and on the sum of gradients $g_0, \dots, g_{t-1}$, but not on their order. 
Moreover, we can always choose from the set of $2n$ gradients $\{\pm A_i: i \in [n]\}$ at each step. 
Thus applying mirror descent to the function $f_U$ for all possible $U$ with the same starting point $U_0$, the total number $N(U_0)$ of points visited in $T := n$ iterations satisfies 
\begin{align*}
N(U_0) \le \sum_{t=0}^n {{t+2n-1} \choose {2n-1}} \le (n+1) \cdot {3n \choose n} \le 2^{O(n)}.
\end{align*}
Since $|T_0| \le 2^{O(n)}$, we obtain a set of $2^{O(n)}$ points $\mathcal{U}$ such that for every $Y = \mathcal{A}(U) \in K'$, there exists some $\tilde{U} \in \mathcal{U}$ so that $\|\mathcal{A}(\tilde{U}) - \mathcal{A}(U)\|_\infty = f_U(\tilde{U}) = f_U(\tilde{U}) - f_U(U) \lesssim \sqrt{D_\Phi^{\max}/n}.$ 
\end{proof}

In the Schatten Norm Setup, we shall pick $K' = K^\circ$ and $T_0 = \{0\}$, i.e. $U_0$ is always $0$. 
For the Spectraplex Setup, we carefully choose a set of starting points $|T_0| \leq 2^{O(n)}$ which has small $D_\Phi^{\max}$ with respect to $K'=\{\mathcal{A}(U): U \in \mathcal{S}_m\}$. 
Since $D_{\Phi}( X \| Y)$ is the quantum relative entropy between $X$ and $Y$ in the Spectraplex Setup, we shall refer to the set of starting points $T_0$ as a (quantum) relative entropy net for $\mathcal{S}_m$. 

\begin{definition}[Quantum Relative Entropy Net] \label{defn:relative_entropy_net}
Given subsets $T, \mathcal{M} \subseteq \mathcal{S}_m$, $T$ is a relative entropy net of $\mathcal{M}$ with error $\varepsilon$ if for any $X \in \mathcal{M}$, we can find $Y \in T$ such that $S(X \| Y) \leq \varepsilon$. 
\end{definition}

\subsection{Initialization for Spectraplex Setup: Relative Entropy Net}
\label{sec:entropy_net}

We start with the following lemma which constructs a relative entropy net on $\mathcal{S}_m$ from an operator norm net.

\begin{lemma}[Relative Entropy Net from Operator Norm Net] \label{lem:entropy_from_op}
Let $X, Y \in \mathcal{S}_m$ satisfies $\|X - Y\|_\op \leq \varepsilon$ for some $\varepsilon \geq 1/m$. Then $S(X \| Y') \leq \log (2m\varepsilon)$, where $Y' := \frac{1}{2}(Y + \frac{I_m}{m}) \in \mathcal{S}_m$. 
\end{lemma}

\begin{proof}
Recall that $\log(\cdot)$ is operator monotone and note that $X \preceq Y + \varepsilon I_m$. 
We then have
\begin{align*}
S(X \| Y') 
& = \tr(X \cdot (\log X - \log Y')) \\
& \leq \tr(X \cdot (\log (Y + \varepsilon I_m) - \log Y')) \\
& \leq \tr(X) \cdot \| \log (Y + \varepsilon I_m) - \log Y' \|_\op \\
& \leq \log \left(2 \cdot \left\|\frac{Y + \varepsilon I_m}{Y + \frac{I_m}{m}} \right\|_\op \right) \leq \log (2 m  \varepsilon) ,
\end{align*}
where the first inequality follows from the operator monotonicity of $\log(\cdot)$, the second follows from matrix H\"{o}lder, and the last follows because $\varepsilon \geq 1/m$ and $\|Y\|_\op \leq 1$. 
\end{proof}

Using the lemma above, we give the following construction for relative entropy nets on $\mathcal{S}_m$.

\begin{theorem}[Entropy Net for Spectraplex] \label{thm:entropy_net}
Given positive integers $h, m$ and $n$ such that $m/h$ is an integer, let $\mathcal{S}_m^{h} \subseteq \mathcal{S}_m$ be the set of $m \times m$ block diagonal matrices on the spectraplex with block size $h \times h$. 
Then we can find a relative entropy net $T$ for $\mathcal{S}_m^{h}$ with error at most $\max(1, \log(2h m/n))$ and size $|T| \leq 2^{O(n)}$. 
\end{theorem}

\begin{proof}
By merging blocks as needed, we may assume $hm \ge n$. By \Cref{lem:entropy_from_op}, it suffices to find an operator norm net $T'$ with size $|T'| \leq 2^{O(n)}$ and distance $\varepsilon = \frac{\max\{h, \log(m/hn)\}}{n}$. 
Let $\ell := m/h$ be the number of blocks, $X_1, \dots, X_\ell \in \mathbb{R}^{h \times h}$ denote the blocks of matrix $X \in \mathcal{S}_m^h$, and $N := 2/\varepsilon = 2n/\max\{h, \log(\ell/n)\}$ (we assume that $N$ is an integer). 
Let $Z := \{z \in \mathbb{Z}_{\geq 0}^\ell: \sum_{i=1}^\ell z_i = N\}$, and for each $z \in Z$, we define \[
T_z := \{X \in \mathcal{S}_m^h: \tr(X_i) = z_i/N, \forall i \in [\ell]\}.
\]
It follows from a standard rounding argument that for any matrix $X \in \mathcal{S}_m^h$, one can find a matrix $Y \in \cup_{z \in Z} T_z$ with $\|X-Y\|_\op \le 1/N = \varepsilon/2$.

We first show that $|Z| \leq 2^{O(n)}$. When $\ell \le 2n$, we have
\begin{align*}
|Z| \leq \binom{N + \ell}{\ell} \leq  \binom{N+2n}{2n} \le \binom{\frac{2n}{h} + 2n}{2n} \le 2^{O(n)}. 
\end{align*}
When $\ell \ge 2n \ge N$, we can bound
\begin{align*}
|Z| \leq \binom{N + \ell}{N} \leq \binom{2\ell}{N} \le \binom{2\ell}{\frac{2n}{\log(\ell/n)}} \le \left( \frac{e \ell \log(\ell/n) }{n} \right)^{\frac{2n}{\log(\ell/n)}} \leq 2^{O(n)} .
\end{align*}

It therefore suffices to construct an $\varepsilon/2$-operator norm net for each $T_z$. 

Fix an arbitrary $z \in Z$. Note that the $i$th block of the matrices in $T_z$ comes from $\frac{z_i}{N} \cdot \mathcal{S}_h$. 
Pick $n_i := z_i h$, we have from \Cref{thm:entropy_number_Schatten} that
\begin{align*}
\mathcal{N}\Big(\frac{z_i}{N} \mathcal{S}_h, \frac{z_i}{N} \cdot \frac{h}{n_i} B_{\op}^h\Big) = \mathcal{N}\Big(\mathcal{S}_h, \frac{h}{n_i} B_{\op}^h\Big) \leq 2^{O(n_i)}. 
\end{align*}
We denote this net as $\widetilde{T}_{z,i}$.  It follows from the above that for any $X_i \in \frac{z_i}{N} \mathcal{S}_h$, there exists $Y_i \in \widetilde{T}_{z,i}$ with $\| X_i - Y_i \|_\op \leq \frac{z_i}{N} \cdot \frac{h}{n_i} = \varepsilon/2$. 
Define $\widetilde{T}_z := \{\diag(Y_1, \dots, Y_\ell): Y_i \in \widetilde{T}_{z,i} \ \forall i \in [\ell]\}$. 
Then for any $X \in T_z$, there exists $Y \in \widetilde{T}_z$ such that $\| X - Y \|_\op \leq \varepsilon/2$, and thus $\widetilde{T}_z$ is indeed an $\varepsilon/2$-operator norm net for $T_z$. 
Furthermore, the size of $\widetilde{T}_z$ can be upper bounded as
\begin{align*}
|\widetilde{T}_z| \leq \prod_{i \in [\ell]} 2^{O(n_i)} = 2^{O(\sum_{i=1}^n z_i h)} = 2^{O(N h )} \leq 2^{O(n)} ,
\end{align*}
since $N \le 2n/h$. 
This proves that $\widetilde{T} := \cup_{z \in Z} \widetilde{T}_z$ is an $\varepsilon$-operator norm net for $\mathcal{S}_m^h$ and has size at most $|\widetilde{T}| \leq 2^{O(n)}$, where we recall that $\varepsilon = \frac{\max\{h, \log(m/hn)\}}{n}$. 
Finally, invoking \Cref{lem:entropy_from_op}, $\widetilde{T}$ can be transformed into a relative entropy net $T$ with size $|T| \leq 2^{O(n)}$ and error at most $\log(2 m \varepsilon) \leq \log(2hm/n)$. 
This finishes the proof of the theorem. 
\end{proof}

\section{Applications of the Spectraplex Setup}
\label{sec:weak_matrix_spencer}

In this section, we prove our matrix Spencer bound for block diagonal matrices in \Cref{thm:block_diagonal_matrix_spencer}, which we restate below. 

\blockdiagonal*

\begin{proof}[Proof of \Cref{thm:block_diagonal_matrix_spencer}] 
By \Cref{thm:entropy_net}, we can find a relative entropy net $T_0$ of $\mathcal{S}_m^h$ with error $D_\Phi^{\max} := \max(1,\log(2hm/n))$ and size $|T_0| \leq 2^{O(n)}$. 
Then using \Cref{lem:size_covering} with the Spectraplex Setup for $K':= \mathcal{A}(\mathcal{S}_m^h)$ and $T_0$ being the relative entropy net, we obtain
\begin{align*}
\mathcal{N} \Big(K', \frac{t}{n} B_\infty^n \Big) \leq 2^{O(n)} ,
\end{align*}
where $t = \sqrt{n\max(1,\log(2hm/n))}$. 
Let $\mathbb{S}_m^h$ be the set of $m \times m$ symmetric block diagonal matrices with block size $h \times h$. 
Define convex body $K'' := \mathcal{A}(B_{S_1}^m \cap \mathbb{S}_m^h \cap \mathbb{S}^m_+)$.
We first prove that $\mathcal{N}(K'', \frac{t}{n}B_\infty^n) \leq 2^{O(n)}$.
Since $\mathcal{N}(K', \frac{t}{n}B_\infty^n) \leq 2^{O(n)}$ by \Cref{thm:entropy_net}, we also have $\mathcal{N}(\frac{j}{n^2} K', \frac{t}{n}B_\infty^n) \leq 2^{O(n)}$ for each integer $j \in [n^2]$. 
We let $H_j$ be the set of centers for the minimum covering of $\frac{j}{n^2} K'$ by translates of $\frac{t}{n}B_\infty^n$ and define $H = \cup_{j \in [n^2]} H_j$. 
Since $|H_j| \leq 2^{O(n)}$, it follows that $|H| \leq 2^{O(n)}$.  
For each $X \in B_{S_1}^m$ that satisfies $X \succeq 0$, we let $\frac{j}{n^2}$ be the multiple of $\frac{1}{n^2}$ that is closest to $\tr(X)$, and set $X':= \frac{j}{n^2 \tr(X)} \cdot X$. 
Then we have
\begin{align*}
\| \mathcal{A}(X') - \mathcal{A}(X) \|_\infty  
\leq \frac{1}{n^2} \cdot \| \mathcal{A}(X) \|_\infty  
\leq \frac{t}{n} .
\end{align*}

As $\tr(X') = \frac{j}{n^2}$, we can also find $Y \in H_j$ with $\| \mathcal{A}(X') - Y \|_\infty \leq \frac{t}{n}$. Therefore, $\| \mathcal{A}(X) - Y \|_\infty \leq \frac{2t}{n}$, and it follows that $K'' \subseteq H + \frac{2t}{n} B_\infty^n$. 
This implies $\mathcal{N}(K'', \frac{t}{n}B_\infty^n) \leq 2^{O(n)}$.

Next note that the dual discrepancy body $K^\circ := \mathcal{A}(B_{S_1}^m) = \mathcal{A}(B_{S_1}^m \cap \mathbb{S}_m^h)$ since each $A_i \in \mathbb{S}_m^h$. 
We have $K^\circ = K'' - K''$, so using \Cref{lem:symmetrization_covering} we get $\mathcal{N}(K^\circ, K'') \leq 2^{O(n)}$. Thus
\begin{align*}
\mathcal{N}\Big(K^\circ, \frac{t}{n}B_\infty^n \Big) \leq \mathcal{N}(K^\circ, K'') \cdot \mathcal{N}\Big(K'', \frac{t}{n}B_\infty^n \Big) \leq 2^{O(n)} ,
\end{align*}

and $\gamma_n(tK) \geq 2^{-O(n)}$ by using \Cref{lem:dual_cov}. 
\Cref{cor:full_coloring} then gives a full coloring $x \in \{\pm 1\}^n$ with discrepancy $\|\sum_{i=1}^n x_i A_i \|_\op \leq O(t)$.
This finishes the proof of the theorem.
\end{proof}

The analysis above also shows that if we can improve the bound in \Cref{thm:entropy_net} to $O(\log (m/n))$ for any block size $h$, then the matrix Spencer conjecture is true. 

\betterentropynet*

\section{Matrix Discrepancy for Schatten Norms}
\label{sec:matrix_discrepancy_schatten}

In this section, we prove the following generalization of \Cref{thm:weak_matrix_spencer} for arbitrary Schatten norms by using a different regularizer for mirror descent.

\matrixdisc*

We first use mirror descent to prove the following covering lemma.

\begin{lemma} \label{lem:spsq_cover}
Let $m \geq \sqrt{n}$, $2 \le p \le q < \infty$, $k := \min(1, m/n)$, $t :=\sqrt{(p-1) n} \cdot k^{1/p-1/q}$ and $q^* := q/(q-1)$. Given symmetric matrices $A_1, \dots, A_n \in \mathbb{R}^{m \times m}$ with $\|A_i\|_{S_p} \leq 1$, we have
\begin{align*}
\mathcal{N}\Big(\mathcal{A}(B^m_{S_{q^*}}), \frac{t}{n} B^n_\infty \Big) \le 2^{O(n)}.
\end{align*}
\end{lemma}

\begin{proof} Denote $p^* := p/(p-1)$. \Cref{thm:entropy_number_Schatten} implies $\mathcal{N}(\mathcal{A}(B^m_{S_{q^*}}), k^{1/q^* - 1/p^*} \mathcal{A}(B^m_{S_{p^*}})) \leq 2^{O(n)}$, so it suffices to show
\begin{align*}
\mathcal{N}\Big(\mathcal{A}(B^m_{S_{p^*}}), \sqrt{\frac{p-1}{n}} B^n_\infty\Big) \le 2^{O(n)}.
\end{align*}{}
This is a direct consequence of \Cref{lem:size_covering} with norm $\| \cdot \|_{S_{p^*}}$, as the Bregman divergence is $D_\Phi(U,0) = \Phi(U) \le \frac{1}{2(p^*-1)} = \frac{p-1}{2}$ for $\|U\|_{S_{p^*}} \le 1$.
\end{proof}

\Cref{lem:spsq_cover} together with \Cref{lem:dual_cov} immediately gives the following weaker measure bound, which we then bootstrap to prove the stronger bound in \Cref{thm:lowrank_matrix_discrepancy}.  

\begin{corollary} \label{cor:weaker_Sp_Sq}
Let $m \geq \sqrt{n}$, $2 \le p \le q < \infty$ and $k := \min(1, m/n)$. 
Given symmetric matrices $A_1, \dots, A_n \in \mathbb{R}^{m \times m}$ with $\|A_i\|_{S_p} \leq 1$, define the convex body 
\begin{align*}
K:= \Big\{x \in \setR^n : \Big\| \sum_{i=1}^n x_i A_i \Big\|_{S_q} \le 1\Big\}. 
\end{align*}
Then $\gamma_n(\sqrt{(p-1) n} \cdot k^{1/p-1/q} \cdot K) \ge 2^{-O(n)}$.
\end{corollary}

\begin{proof} [Proof of \Cref{thm:lowrank_matrix_discrepancy}]

Let $p_0 := \max(2,\log(2 rk))$. For $p \le p_0$ the result follows directly from \Cref{cor:weaker_Sp_Sq}, so we may assume $p \ge p_0$. Also note that we may assume $rk \ge 1$ since we can increase smaller values of $r$ without changing the bound on the right side. Remark that $\|A_i\|_{S_{p_0}} \le r^{1/p_0-1/p} \|A_i\|_{S_p} \le  r^{1/p_0-1/p}$ since the matrices have rank at most $r$. \Cref{cor:weaker_Sp_Sq} then implies that the convex body
\begin{align*}
\sqrt{p_0 n} \cdot k^{1/p_0-1/q} \cdot r^{1/p_0-1/p} \cdot K
\end{align*}
has Gaussian measure $2^{-O(n)}$. Since $\sqrt{p_0 n} \cdot k^{1/p_0-1/q} \cdot r^{1/p_0-1/p} \lesssim \sqrt{p_0 n} \cdot k^{1/p-1/q}$ by the choice of $p_0$, it follows that 
\begin{align*}
\gamma_n (\sqrt{n \max(1, \log(rk))} \cdot k^{1/p-1/q} \cdot K) \ge 2^{-O(n)},
\end{align*} 
so that  \Cref{thm:partial_coloring} and \Cref{cor:full_coloring} yield the partial coloring and full coloring, respectively. The factor $(1/2 + 1/p - 1/q)^{-1}$ comes from the contribution of the exponent of $n$ in the geometric sum, analogous to the second part of \Cref{cor:full_coloring}.
\end{proof}

\section{Lower Bound Examples for Matrix Discrepancy} \label{sec:examples}

In this section, we give a few examples to illustrate the tightness of our results in \Cref{thm:lowrank_matrix_discrepancy} for various regimes of the dimension $m$ and rank $r$ of the input matrices. 


\subsection{Low Dimension Regime of $m = \Theta(\sqrt{n})$}  \label{sec:lower_bound_low_dimension}
In the regime of $m = \Theta(\sqrt{n})$, we have $k = \min(1,m/n) = \Theta(1/\sqrt{n})$ and $r \leq O(\sqrt{n})$ and our partial coloring bound in \Cref{thm:lowrank_matrix_discrepancy} is thus $O(n^{1/2 + 1/2q - 1/2p})$. 
This bound is tight up to constants due to the following  example\footnote{Thanks to Aleksandar Nikolov for suggesting this construction.}. 

\begin{lemma}[Example: $m = \sqrt{n}$]
Let $m = \sqrt{n}$ be a power of $2$, and $2 \leq p \leq q \leq \infty$. There exist matrices\footnote{These matrices can easily be made symmetric in $\setR^{2m \times 2m}$.} $A_1, \dots, A_n \in \mathbb{R}^{m \times m}$ with $\|A_i\|_{S_p} \leq 1$ such that $\|\sum_{i=1}^n x_i A_i\|_{S_q} \gtrsim n^{1/2 + 1/2q - 1/2p}$ for any partial coloring $x \in \{\pm 1\}^n$ with $|\{i: |x_i| = 1\}| \geq n/2$.
\end{lemma} 
\begin{proof}
The idea is to construct an orthogonal basis on $\mathbb{R}^{m \times m}$ with  $\|A_i\|_F^2 = m$.
Let $H \in \mathbb{R}^{m \times m}$ be the Walsh-Hadamard matrix, and $D_1, \dots, D_m$ be diagonal matrices with $(D_i)_{j, j} := H_{i,j}$. 
Let $P_1, \dots, P_m$ be disjoint permutation matrices, i.e. each $P_i$ permutes the standard orthonormal basis $\{e_1, \dots, e_m\}$ and each pair $P_i$, $P_j$ have disjoint non-zero entries.
For instance, we may take $(P_i)_{j,k} := 1$ if $j-k \equiv i \mod m$ and $0$ otherwise. 
We then define the $n$ matrices $A_{i+mj} := D_i P_j$ for $i, j \in [m]$. 
Note that these matrices form an orthogonal basis of $\setR^{m \times m}$, so for any partial coloring $x \in \{\pm 1\}^n$ with $|\{i: |x_i| = 1\}| \geq n/2$, we have
\[
\Big\|\sum_{i=1}^n x_i A_i \Big\|^2_F = \tr\left(\left(\sum_{i=1}^n x_i A_i \right)^2 \right) = m \cdot \sum_{i=1}^n x_i^2 \geq m n /2. 
\]
By H\"older's inequality, this implies that
\begin{align*}
\Big \| \sum_{i=1}^n x_i A_i \Big \|_{S_q} \geq m^{1/q - 1/2} \cdot \Big \| \sum_{i=1}^n x_i A_i \Big \|_F \gtrsim n^{1/2 + 1/2q} .
\end{align*}

Also note that each matrix $A_i$ has all singular values equal to $1$, and therefore $\|A_i\|_{S_p} = m^{1/p} = n^{1/2p}$. Scaling the matrices $A_i$ down by a factor of $n^{1/2p}$ proves the lemma. 
\end{proof}

\subsection{Rank-$1$ Matrices and $m \geq n$} 
\label{sec:lower_bound_rank_1}
In the regime of $r = 1$ and $m \geq n$, we may assume wlog that $p=2$.
Then the discrepancy bound in \Cref{thm:weak_matrix_spencer} is $O(\sqrt{n})$. 
This bound is again tight up to a constant factor. 

\begin{lemma}[Example: $r=1$ and $m=n$] \label{lem:rank_1} Let $2 \leq q \leq \infty$.
There exist symmetric rank-$1$ matrices $A_1, \dots, A_n \in \mathbb{R}^{n \times n}$ with $\|A_i\|_F \leq 1$ such that any partial coloring $x \in [-1,1]^n$ with $|\{i: |x_i| = 1\}| \geq n/2$ satisfies
\begin{align*}
\Big \|\sum_{i=1}^n x_i A_i \Big \|_{S_q} \geq \Big \|\sum_{i=1}^n x_i A_i \Big \|_\op \gtrsim \sqrt{n} .
\end{align*}
\end{lemma}

\begin{proof}
For each $i \in [n-1]$, we define the rank-$1$ matrices $A_i := \frac{1}{2} (e_i + e_n)(e_i + e_n)^\top$ for $i \in [n]$, where $e_i \in \mathbb{R}^n$ is the unit vector with a single $1$ in the $i$th coordinate and $0$ elsewhere, and $A_n = 0$. 
Note that each $\|A_i\|_F = 1$ by definition. 
For any partial coloring $x \in [-1,1]^n$ with $|\{i: |x_i| = 1\}| \geq n/2$, we have
\begin{align*}
\sum_{i=1}^n x_i A_i = \frac{1}{2} \cdot \left( 
\begin{matrix}
x_1 & 0 & \cdots & 0 &  x_1 \\
0 & x_2 & \cdots & 0 & x_2 \\
\vdots & \vdots  & \ddots & \vdots & \vdots \\
0 & 0 & \cdots & x_{n-1} & x_{n-1}\\
x_1 & x_2 & \cdots & x_{n-1} & \sum_{i=1}^{n-1} x_i 
\end{matrix}
\right) .
\end{align*}

It then follows that 
\begin{align*}
\Big\|\sum_{i=1}^n x_i A_i \Big\|_\op \ge \Big\|\sum_{i=1}^n x_i A_i e_n \Big\|_2 \gtrsim \sqrt{n}.
\end{align*}
This completes the proof of the lemma. 
\end{proof}

As an immediate corollary of \Cref{lem:rank_1}, we obtain an $\Omega(\sqrt{n})$ lower bound for matrix Spencer when $m = n$ and all matrices are rank-$1$. 

\RankOneMatrixSpencer*

Another immediate consequence of \Cref{lem:rank_1} is a lower bound of $\Omega(\sqrt{\min(m,n)})$ for Schatten-$2$ to operator norm discrepancy, which is the generalization of the Koml\'os problem to matrices. This shows that the Koml\'os conjecture, which states that the $\ell_2$ to $\ell_\infty$ vector discrepancy is upper bounded by a universal constant, cannot be true for matrices. 

\MatrixKomlos*

\section*{Acknowledgements}
We thank Aleksandar Nikolov for insightful discussions and for suggesting the lower bound example when $m = \sqrt{n}$. 
We thank Thomas Rothvoss, Sivakanth Gopi and Mehtaab Sawhney for helpful discussions.

\appendix

\section{An Application of Banaszczyk's Theorem}
\label{sec:banaszczyk_matrix}

We give an alternative simpler proof of the $O(m^{1 + 1/q - 1/p})$ bound for $S_p$ to $S_q$ matrix discrepancy when $m = O(\sqrt{n})$ using the following theorem of Banaszczyk \cite{b98}. 

\begin{theorem}[Banaszczyk \cite{b98}] \label{thm:ban}
Let $K \subseteq \mathbb{R}^m$ be a convex body with $\gamma_m(K) \geq 1/2$. Then for any vectors $v_1, \dots, v_n \in \mathbb{R}^m$ with $\|v_i\|_2 \leq 1$, there exists $x \in \{\pm 1\}^n$ such that $\sum_{i=1}^n x_i v_i \in 5 K$.  
\end{theorem}

Applying \Cref{thm:ban} to a suitable scaling of the operator norm ball immediately gives the following matrix discrepancy bound. 

\begin{corollary} \label{thm:ban_matrix_disc}
Let $2 \leq p \leq q \leq \infty$. Given matrices $A_1, \dots, A_n \in \mathbb{R}^{m \times m}$ with $\|A_i \|_{S_p} \leq 1$, there exists $x \in \{\pm 1\}^n$ such that $\|\sum_{i=1}^n x_i A_i\|_{S_q} \lesssim m^{1 + 1/q -  1/p}$. 
\end{corollary}

\begin{proof}
Note that $\|A_i\|_{S_p} \leq 1$ implies $\|A_i\|_{S_2} \leq m^{1/2 - 1/p}$. It is well-known that $\gamma_m(4 m^{1/2} \cdot B_{\op}^m) \geq 1/2$ (see Theorem 7.3.1 of \cite{v18book}).  
Thus, \Cref{thm:ban} yields some $x \in \{\pm 1\}^n$ such that $\sum_{i=1}^n x_i A_i\in O(m^{1 - 1/p}) \cdot B_\op^m$. It follows that $\|\sum_{i=1}^n x_i A_i \|_{S_q} \leq O(m^{1 + 1/q - 1/p})$. 
\end{proof}

\begin{corollary}[Matrix Koml\'os] 
Given matrices $A_1, \dots, A_n \in \mathbb{R}^{m \times m}$ with $\|A_i \|_F \leq 1$, there exists $x \in \{\pm 1\}^n$ such that $\|\sum_{i=1}^n x_i A_i\|_{S_q} \lesssim \sqrt{\min(m,n)}$, matching the lower bound in \Cref{komlos_lower}.
\end{corollary}

\begin{proof}
It suffices to take the best between a random coloring, which has discrepancy $O(\sqrt{n})$, and that of \Cref{thm:ban_matrix_disc}.
\end{proof}

\bibliographystyle{alpha}
\bibliography{bib.bib}

\end{document}